\documentclass{amsart}
\usepackage[utf8]{inputenc}

\usepackage{color}
\usepackage{mathrsfs}

\usepackage[shortlabels]{enumitem}
\setlist[enumerate, 1]{i)}
\usepackage[T1]{fontenc}
\usepackage{pifont}
\usepackage{amsmath}
\usepackage{amssymb}
\usepackage{amsthm}
\usepackage{amscd}
\usepackage{bm}

\newtheorem{theorem}{Theorem}[section]
\newtheorem{lemma}[theorem]{Lemma}
\newtheorem{proposition}[theorem]{Proposition}

\theoremstyle{remark}
\newtheorem{remark}{Remark}

\newcommand\numberthis{\addtocounter{equation}{1}\tag{\theequation}}

\newcommand{\N}{\mathbb{N}}
\newcommand{\Z}{\mathbb{Z}}
\newcommand{\R}{\mathbb{R}}
\newcommand{\C}{\mathbb{C}}
\newcommand{\Db}{\mathbb{D}}
\newcommand{\s}{\mathbb{S}}

\newcommand{\id}{\,\mathrm{d}}

\newcommand{\supp}{\operatorname{supp}}

\newcommand{\abs}[1]{\lvert#1\rvert}
\newcommand{\norm}[1]{\lVert#1\rVert}
\newcommand{\bvec}[1]{\boldsymbol{#1}}

\newcommand{\innorm}[1]{\left| #1 \right|_{\C^2}}
\newcommand{\dotprod}[2]{
  \bvec{#1}\mkern1mu{\cdot}\mkern1mu\bvec{#2} \,}
  
\newcommand{\inprodS}[2]{\left ( #1 , #2\right )_{\C^2}}  
\newcommand{\Scalprod}[2]{\left \langle #1 , #2\right \rangle }  

\newcommand{\Dom}[1]{\mathcal{D}( #1)}

\newcommand{\sigvec}{\boldsymbol{\sigma}}

\newcommand{\D}{D}
\newcommand{\T}{T}
\newcommand{\Trace}{\gamma}
\newcommand{\KC}{K}
\newcommand{\KCbar}{\overline{K}}

\newcommand{\Bound}{{\partial \Omega}}

\title[ ]{Self-adjointness of two-dimensional Dirac operators
 on domains} 
\author[Benguria]{Rafael D. Benguria}
\author[Fournais]{S\o ren Fournais} \author[Stockmeyer]{Edgardo
  Stockmeyer} \author[Van Den Bosch]{Hanne Van Den Bosch}
\address{ Rafael D. Benguria, Edgardo Stockmeyer and Hanne Van Den Bosch\\ Instituto de F\'\i sica\\
  Pontificia Universidad Cat\'olica de Chile\\
  Vicu\~na Mackenna 4860\\
  Santiago 7820436, Chile.}  \address{S\o ren Fournais, Department of
  Mathematics, Aarhus University, Ny Munkegade 118, DK-8000 Aarhus,
  Denmark}
\begin{document}

\begin{abstract}
We consider Dirac operators defined on planar domains. For a large
class of boundary conditions, we give a direct proof of their self-adjointness in the
Sobolev space $H^1$.
\end{abstract}

\maketitle

\section{Introduction} 
We consider a massless two-dimensional Dirac operator on a bounded domain
$\Omega \subset \R^2$ with $C^2$-boundary $\partial \Omega$.  Choosing
appropriate units, the Dirac operator acts as the differential
expression
\[
T:= -i \dotprod{\sigma}{} \nabla  = 
\begin{pmatrix}
                                        0 &1\\1&0
                                       \end{pmatrix} (-i \partial_1) 
                                       + \begin{pmatrix}
                                        0 &-i\\ i & 0
                                       \end{pmatrix} (-i \partial_2). 
\]
We denote by $\D_\eta$  the operator acting as $\T$
on functions in the domain 
\[
\Dom{\D_\eta} :=  \{ u \in H^1(\Omega, \C^2)\,|\, P_{-,\eta} \Trace u = 0\}.
\]
Here $\Trace$ is the trace operator on the boundary of $\Omega$ and
the orthogonal projections $P_{\pm, \eta}$ are defined as
\[
P_{\pm,\eta} = \frac12(1 \pm A_\eta ),      
\quad A_\eta = \cos(\eta) \dotprod{\sigma}{t}+ \sin(\eta) \sigma_3,
\]
where $\bvec t$ is  the unit vector tangent to the boundary and $\eta$ is
a real function on the boundary. 

In the physics literature operators of this type were first considered
in 1987 by Berry and Mondragon to study two-dimensional neutrino
billards \cite{BerryMondragon}. In recent years, they have gained
renewed interest due to their applications in the description of
graphene quantum dots and nano-ribbons (see
e.g. \cite{CastroNetoetAl,AkhmerovBeenakker,beneventano} and
references therein). The most commonly used boundary conditions are
those when $\eta\in\{0,\pi\}$  and $\eta\in\{\pi/2,3\pi/2\}$, known as
 infinite mass and zigzag boundary conditions, respectively.

Using integration by parts and the hermiticity of the Pauli matrices,
it is straightforward to check that $\D_\eta$ is a
symmetric operator. We have, for all $u, v \in H^1(\Omega, \C^2)$,
\begin{align*}
 \Scalprod{u}{\T v} &= \int_\Omega -i \inprodS{u}{\dotprod{\sigma}{} \nabla v } \\
	  & = \int_\Omega -i \nabla\cdot {}\inprodS{u}{\sigvec v} +i \int_\Omega \inprodS{\dotprod{\sigma}{} \nabla  u}{v}\numberthis \label{eq : ibp}\\
	  & = \Scalprod{\T u}{ v} - i \int_\Bound \inprodS{u}{\dotprod{n}{\sigma} v}, 
\end{align*}
where $\bvec n$ is the outward normal vector to $\Bound$.
If $u, v \in \Dom{\D_\eta}$, the boundary term cancels since the
anticommutation relations of the Pauli matrices imply 
\begin{equation}\label{referee1}
\{A_\eta,\dotprod{n}{\sigma} \} = 0.
\end{equation}

To determine when $\D_\eta$ is actually self-adjoint, in the case of
$C^\infty$-boundaries, one may adapt the corresponding theorems of
\cite{BoossLeschZhu} to our case (see for instance
\cite{Prokhorova2013}). However, the operators treated in
\cite{BoossLeschZhu} are more general and the proofs require
sophisticated techniques from the analysis of pseudodifferential
operators. Our proof, given in Section \ref{sec : s-a}, is simpler and
also works in cases with limited regularity of $\eta$ and $\partial
\Omega$.
\begin{theorem} \label{thm : self-adjoint}
Given $\Omega \subset \R^2$, bounded, with $C^2$-boundary, and $\eta \in C^1(\Bound)$, define $\D_\eta$ as above.
If $\cos \eta(s) \neq 0$ for all $s \in \Bound$, then $\D_\eta$ is self-adjoint on $\Dom{\D_\eta}$. 
\end{theorem}

\begin{remark}
Our proof of self-adjointness is really an elliptic regularity result for the Dirac system. We implicitly establish the following inequality: 

Suppose that $\Omega$ and $\eta$ satisfy the conditions of Theorem~\ref{thm : self-adjoint}. Then there exists a constant $C>0$ such that
\begin{align}
\| u \|_{H^1(\Omega)} \leq C \left( \|u \|_{L^2(\Omega)} + \|Tu \|_{L^2(\Omega)} \right),
\end{align}
for all $u \in L^2(\Omega,\C^2)$ satisfying the boundary condition $P_{-,\eta} \Trace u =0$.
Notice that we establish below that the boundary trace $\Trace u$
exists in $H^{-1/2}(\Bound)$ if $u, Tu \in L^2(\Omega,\C^2)$. 
\end{remark}

\begin{remark}
 We do not know whether the hypothesis $\cos \eta (s) \neq 0$ is optimal, but it can not be relaxed much. If $\D_\eta$ is self-adjoint on a domain contained in $H^1(\Omega, \C^2)$,
it follows from the compact embedding of $H^1(\Omega) \subset L^2(\Omega)$ that its resolvent is compact.
Thus, the spectrum of $\D_\eta$ consists of eigenvalues of finite
multiplicity accumulating only at $\pm \infty$. This is to be
contrasted with the case of zigzag boundary conditions, $\cos\eta=0$,
which has $0$ in the essential spectrum. In particular, the
corresponding operator is not self-adjoint on a domain included in
$H^1(\Omega, \C^2)$ (see \cite{schmidt1994,freitassiegl2014}).
More generally, we show in the appendix that, if  $\cos \eta (s)$ tends to zero at least quadratically when $s \to s_0 \in \Bound$, 
$0 \in \sigma_\mathrm{ess}(\D_\eta)$.
\end{remark}

The rest of the paper presents the proof of Theorem~\ref{thm : self-adjoint}. 
Our strategy is to show directly that $\Dom{D_\eta^*} \subset \Dom{D_\eta}$.
The difficult part is showing the inclusion $\Dom{D_\eta^*}\subset H^1(\Omega, \C^2)$, 
for which it is necessary to prove the regularity of the boundary values of functions in $\Dom{D_\eta^*}$.
This step exploits the interplay between the projections giving the boundary conditions and the special structure of the Dirac operator.
We first establish the necessary results when $\Omega = \mathbb{D}$, the unit disc.
Finally, the Riemann mapping theorem allows to treat the general case as well.

\section{Self-adjointness} \label{sec : s-a} We first fix some
notations. We  work with spaces of $\C^2$-valued functions such as
$H^1( \Omega, \C^2), C^{\infty}(\Omega, \C^2),\ldots$.  For shortness
of notation, we  often omit the $\C^2$ and just write
$H^1(\Omega), C^{\infty}(\Omega), \ldots$ when no possible confusion
occurs.  We will consider a fixed domain $\Omega$ with
$C^2$-boundary $\Bound$.  We denote by $\bvec{n}(s)$ and $\bvec{t}(s)$
the outward normal and the tangent vector to the boundary at the point
$s \in \Bound$, chosen such that $\bvec n, \bvec t$ is positively
oriented.  If $\bvec t (s) = (t_1(s), t_2(s))$, we define $t(s) =
t_1(s)+ i t_2 (s)$, the tangent vector seen as a number in $\C$.
Associated to the domain $\Omega$ we have the trace operator at the
boundary $\Trace : C^1(\overline{\Omega}) \to C^1(\Bound)$, and an
extension operator $E : C^1(\Bound) \to C^1(\overline{\Omega}) $.  We
recall that $\Trace$ extends to a bounded operator from
$H^{s+1/2}(\Omega)$ to $H^{s}(\Bound)$, and $E$ from $H^{s}(\Bound)$
to $H^{s+1/2}(\Omega)$ for all $s \in (0,2)$.  We denote by
$\mathscr{D}'(\Omega)$ the space of distributions, i.e., the dual of
$C^{\infty}_0(\Omega)$.

We will also consider a fixed function $\eta$ defining the boundary conditions and write simply $D$ for $D_\eta$.

In passing, we recall our definition for the Pauli matrices
\[
\sigma_1 =\begin{pmatrix}
 0&1\\1&0
\end{pmatrix},
\quad
\sigma_2 = \begin{pmatrix}
 0&-i\\i&0
\end{pmatrix},
\quad
\sigma_3 = \begin{pmatrix}
 1&0\\0&-1
\end{pmatrix}.
\]
They satisfy the (anti)commutation relations
 \[
 \{\sigma_j, \sigma_k \} = 2\delta_{jk}, \quad 
 [\sigma_j, \sigma_k]  =  2 i \epsilon_{jkl} \sigma_l, \quad j,k,l \in \{ 1,2,3 \},
 \]
where $\delta_{jk}$ is the Kronecker delta and $\epsilon_{jkl}$ is the Levi-Civita symbol, 
which is totally antisymmetric and normalized by $\epsilon_{123} = 1$.

\subsection{General considerations}
First, we will need some regularity properties of $v \in \Dom{\D^* } $.
\begin{lemma}\label{lemmasf}
  Let $\mathcal{K}:=\{u\in L^2(\Omega,\C^2)\,|\, Tu\in
  L^2(\Omega,\C^2)\}$ equipped with the graph-norm
  $\|u\|_{\mathcal{K}}^2=\|u\|^2+\|Tu\|^2,$ where $T$ acts as a
  differential operator on distributions in $\Omega$. Then
  $\mathcal{K}$ is a Hilbert space and
  $C^\infty(\overline{\Omega},\C^2)$ is dense in $\mathcal{K}$.
\end{lemma}
\begin{proof}
First we show that $\mathcal{K}$ is  a Hilbert space. Take a Cauchy sequence
$(u_n)_{n\in\N}\subset \mathcal{K}$ with $u_n\to u$ and $Tu_n\to v$ in
$L^2$. We have for any test function $\varphi\in C^\infty_0(\Omega)$ 
\begin{align*}
  Tu[\varphi]&=u[T\varphi]=\lim_{n\to\infty}\Scalprod{u_n}{T\varphi}=
\lim_{n\to\infty}Tu_n[\varphi]=\lim_{n\to\infty}\Scalprod{Tu_n}{\varphi}=
\Scalprod{v}{\varphi}.
\end{align*}
Therefore, $Tu=v$ and in particular $u\in \mathcal{K}$.

Recall that by definition
$u \in C^{\infty}(\overline{\Omega})$ iff $u$ is the restriction to
$\Omega$ of a smooth function (spinor) on ${\mathbb R}^2$.
To prove the density of $C^{\infty}(\overline{\Omega})$ it suffices to
show that if
\begin{align}\label{gg1}
\Scalprod{v}{u}_\mathcal{K}=\Scalprod{v}{u}+\Scalprod{Tv}{Tu}=0,
\end{align}
for all $u \in C^{\infty}(\overline{\Omega})$ then $v$ vanishes.
Let $w:=Tv \in L^2(\Omega)$. It follows from \eqref{gg1} that 
\begin{align}\label{gg2}
Tw=-v\quad \mbox{in}\quad \mathscr{D}'(\Omega).
\end{align}
 Define $\widetilde{v}$ and $\widetilde{w}$ as the extensions
by zero to $L^2(\R^2)$ of $v$ and $w$, respectively. For any
$\varphi\in C^\infty_0(\R^2)$ we calculate using \eqref{gg1}
\begin{align*}
  T\widetilde{w}[\varphi]=
  \Scalprod{\widetilde{w}}{T\varphi}_{L^2(\R^2)}=
 \Scalprod{{w}}{T\varphi}_{L^2(\Omega)}=
   \Scalprod{-v}{\varphi}_{L^2(\Omega)}
=\Scalprod{-\widetilde{v}}{\varphi}_{L^2(\R^2)}.
\end{align*}
Therefore, $T\widetilde{w}=-\widetilde{v}\in L^2(\R^2)$. By
ellipticity we find that $\widetilde{w}\in H^1(\R^2)$. Moreover, using
\cite[Proposition IX.18]{brezisbook} we get that ${w}\in
H^1_0(\Omega)$.

Let $(\varphi_n)_{n\in\N}\subset C_0^\infty(\Omega)$ be a sequence
with $\varphi_n\to w$ in $H^1(\Omega)$. For any $u\in \mathcal{K}$ 
\begin{align*}
  \Scalprod{v}{u}_\mathcal{K}&=\Scalprod{v}{u}_{L^2(\Omega)}+\Scalprod{w}{Tu}_{L^2(\Omega)}\\
&=\Scalprod{v}{u}_{L^2(\Omega)}+\lim_{n\to\infty}
\Scalprod{T\varphi_n}{u}_{L^2(\Omega)}\\
&=\Scalprod{v}{u}_{L^2(\Omega)}+
\Scalprod{Tw}{u}_{L^2(\Omega)}=0,
\end{align*}
where the last equality follows from \eqref{gg2} and implies that $v=0$.
\end{proof}
\begin{lemma}\label{lemma : H1loc}
 We have that $\mathcal{D}(\D^*)\subset \mathcal{K}$. Moreover, 
$\mathcal{K}\subset H^1_{\rm loc}(\Omega, \C^2)$.
\end{lemma}
\begin{proof}
  Fix $v \in \Dom{\D^*}$ and define $\tilde v:=\D^* v \in L^2
  (\Omega)$. By definition $\T v$ is a distribution, thus for any $u
  \in C^\infty_0(\Omega)$
 \[
 \T v [u] \equiv \Scalprod{v}{ \T u} = \Scalprod{v}{ \D u} = \Scalprod{\tilde v}{ u},
 \]
 since $C^\infty_0(\Omega) \subset \Dom{\D}$.
 This shows that  the distribution $\T v$ can be identified with the $L^2$-function
 $\D^* v$ and thus $v\in \mathcal{K}$.
 
 Let now $v\in \mathcal{K}$.  By Lemma \ref{lemmasf} we may choose a
 sequence of $C^\infty(\Omega)$-functions $(v_n)_{n\in\N}$ that
 converges to $v$ in $L^2(\Omega)$ and such that $\T v_n$ converges to
 $\T v$ in $L^2(\Omega)$.

 Fix an open set $A$ such that $\bar A \subset \Omega$. 
 We will show that $\nabla v_n$ converges in $L^2 (A)$.
 Take a cut-off function $\chi_A \in C^\infty_0 (\Omega)$ such that $\chi_A = 1$ on $A$.
By equation \eqref{eq : ibp} we have, for all $u  \in
C^\infty_0(\Omega)$, that $\norm {Tu} = \norm{\nabla u}$. Thus we can bound
  \begin{align*}
  \int_A \innorm{\nabla (v_n- v_m)}^2 & \leq  \int_\Omega \innorm{\nabla \chi_A  (v_n-v_m)}^2 \\
			    & = \int_\Omega \innorm{\T \chi_A  (v_n-v_m)}^2\\
			    & \leq \|\nabla \chi_A\|_\infty^2 \norm{v_n-v_m}^2 + \norm{\T  (v_n-v_m)}^2.
\end{align*}
  This finishes the proof.
\end{proof}

By Lemma~\ref{lemma : H1loc} the difficult part in proving the inclusion $\Dom{\D^*}
\subset \Dom{\D}$ is to show regularity of $v\in \Dom{\D^*}$ up to the
boundary.  To do so it is sufficient to prove that $v$ has a
sufficiently regular trace on the boundary $\partial \Omega$.  First
we show that traces exist  as distributions.

\begin{lemma} \label{lemma : H1/2trace} The trace $\Trace$ extends to
  a continuous map \[\Trace: \mathcal{K}\to H^{-1/2}(\partial\Omega,
  \C^2).\] Moreover, if $v \in \Dom{\D^* }$ then $P_- \Trace v = 0$.
  An equivalent formulation of this is that $\Trace v_2 = \frac{1-\sin(\eta)}{\cos(\eta)} t \Trace v_1$.
 \end{lemma}
\begin{proof}
Let $v\in \mathcal{K}$ and let  $(v_n)_{n\in\N}$  be a  $C^\infty(\overline{\Omega})$-sequence
approximating $v$ in the $\mathcal{K}$-norm.
We will show that the traces $\Trace v_n$ of the $v_n$'s
  converge in $H^{-1/2}(\Bound)$. 

  Fix $f \in C^\infty (\Bound )$.  By \cite[Theorem 7.53]{bookAdams}
  it is possible to extend $f$ to a regular function $u \equiv E f$ on
  $\Omega$ satisfying $\Trace u = f$ with $\norm{u}_{H^1( \Omega)} \leq
  C_E \norm{f}_{H^{1/2}( \Bound)}$, with $C_E$ only depending on
  $\Omega$. By the same
  calculation as in \eqref{eq : ibp},
\[
i \int_\Bound (\Trace v_n,\dotprod{\sigma}{n} f) = 
 \Scalprod{\T v_n}{u} - \Scalprod{v_n}{ \T u}.
\]
This shows
\begin{align*}
 \abs{\Scalprod{\Trace(v_n-v_m)}{\dotprod{\sigma}{n} f}_\Bound} 
	      &\leq \norm{\T (v_n-v_m)} \norm{u} + \norm{v_n-v_m} \norm{\nabla u} \\
	      & \leq \bigl( \norm{\T (v_n-v_m)}+ \norm{v_n-v_m}\bigr) \norm{u}_{H^1( \Omega)} \\
	      & \leq C_E \bigl( \norm{\T (v_n-v_m)}+ \norm{v_n-v_m} \bigr) \norm{f}_{H^{1/2}( \Bound)}.
 \end{align*}
 This in turn proves that the limit $\dotprod{\sigma}{n} \Trace v$ of
 $\dotprod{\sigma}{n}\Trace v_n$ exists in $H^{-1/2}(\Bound)$.  Since
 $\dotprod{\sigma}{n}$ is a pointwise invertible matrix (in fact
 $(\dotprod{\sigma}{n})^2=1$) with $C^1$-entries, the same conclusion
 holds for $\Trace v$.

 Assume now that $v\in \mathcal{D}(\D^*)$ and that $u \in \Dom{D}$,
 then  $f:=\Trace u= P_+ f$ and
\[
i \int_\Bound \inprodS{\Trace v}{\dotprod{\sigma}{n} f} = 
\int_\Omega \inprodS{\T v}{u} - \inprodS{v}{ D u} = \Scalprod{D^* v}{ u} - \Scalprod{v}{ D u} = 0.
\]
In addition, using \eqref{referee1} we have that $P_+
\dotprod{\sigma}{n}= \dotprod{\sigma}{n} P_-$. Then, $\inprodS{\Trace
  v}{\dotprod{\sigma}{n} f} =\inprodS{\Trace v}{\dotprod{\sigma}{n}
  P_+f}= \inprodS{\Trace v}{P_-\dotprod{\sigma}{n}
  f}=\inprodS{P_-\Trace v}{\dotprod{\sigma}{n} f}$. Thus, we have
shown that $P_-\Trace v=0$. This finishes the proof.
\end{proof}

The next lemma shows that improving the regularity of the traces is all that is left to do.

\begin{lemma} \label{lemma : boundary->bulk}
 If $v \in \mathcal{K}$ and $\Trace v \in H^{1/2}(\Bound, \C^2)$,
 then $v \in H^1 (\Omega, \C^2)$.
\end{lemma}
\begin{proof}
Let $v \in \mathcal{K}$ with  $\Trace v \in H^{1/2}(\Bound)$. By
replacing $v$ by $v-E(\Trace(v))$, where  
 $E : H^{1/2}(\Bound) \mapsto H^1(\Omega)$ is the (continuous)
 extension operator, it suffices to consider the case when $\Trace
 v=0$. 

Write $w:=Tv\in L^2(\Omega)$. Next we show that
\begin{align}\label{gg3}
  \Scalprod{v}{T\varphi}=\Scalprod{w}{\varphi}, \quad \mbox{for
    all}\quad
\varphi\in C^\infty(\overline{\Omega}, \C^2).
\end{align}
Let $(v_n)_{n\in\N}$ be a  $C^\infty(\overline{\Omega})$-sequence
approximating $v$ in the $\mathcal{K}$-norm. Then by Lemma~\ref{lemma
  : H1/2trace}, $\Trace v_n\to\Trace v=0$ in
$H^{-1/2}(\partial\Omega)$. We calculate for $\varphi\in
C^\infty(\overline{\Omega})$ using \eqref{eq : ibp}
\begin{align*}
  \Scalprod{v}{T\varphi}&=\lim_{n\to\infty}
  \Scalprod{v_n}{T\varphi}=\lim_{n\to \infty}
  \left(\Scalprod{Tv_n}{\varphi}-i\int_{\partial\Omega}
    \inprodS{\Trace v_n}{\dotprod{n}{\sigma} \Trace \varphi}\right)\\
&=\Scalprod{w}{\varphi},
\end{align*}
where the boundary term vanishes since $\Trace\varphi\in
H^{1/2}(\partial\Omega)$. This proves \eqref{gg3}.

Let $\widetilde{v}$ and $\widetilde{w}$ be the extensions by zero  to $L^2(\R^2)$ of
$v$ and $w$, respectively.  Then, by \eqref{gg3} 
\begin{align}
  \label{gg4}
  T\widetilde{v}=\widetilde{w},\quad\mbox{in}\quad \mathscr{D}'(\R^2).
\end{align}
 From this we conclude that $\widetilde{v}\in H^1(\R^2)$ and thus
 $v\in H^1(\Omega)$. This finishes the proof.
 \end{proof}

In order to take advantage of the special structure of the Dirac operator,
it will be convenient to identify $\bvec x \in \R^2$ with the complex number $z = x_1 + i x_2$.
In this notation, the Dirac operator reads
\[
T u (z) = -2 i \begin{pmatrix}
           0 & \partial_z \\ \partial_{z^*} & 0
          \end{pmatrix} u (z) = -2i
          \begin{pmatrix}
           \partial_{z} u_2 (z) \\ \partial_{z^*} u_1 (z)
          \end{pmatrix}          ,
\]
where we introduced the Cauchy-Riemann derivatives 
$\partial_z := \frac{1}{2}(\partial_1 -i \partial_2)$ and $\partial_{z^*} :=  \frac{1}{2}(\partial_1 +i \partial_2)$.
In addition, we introduce the Cauchy kernel
\[
(\KC f) (\zeta)= \frac{1}{2 \pi i} \int_\Bound \frac{f(z)}{z-\zeta} \id z
\]
and its formal conjugate
\[
(\KCbar f) (\zeta)= \frac{-1}{2 \pi i} \int_\Bound \frac{f(z)}{z^*-\zeta^*} \id z^*.
\]
The kernels $\KC, \KCbar$ clearly define operators from $C^\infty(\Bound, \C)$ to $C^\infty(\Omega, \C)$. 
With these definitions we construct an operator on $C^\infty(\Bound, \C^2)$ by setting
\[
S = \begin{pmatrix}
     \KC & 0 \\ 0 & \KCbar
   \end{pmatrix}.
\]
Actually, $-2 \Trace S \dotprod{\sigma}{n}$ coincides with the Calder\'on projector for the Dirac operator as defined for instance in \cite[Chapter 12]{bookBooss}.

\subsection{The Cauchy kernel on the unit circle}

 On the unit circle $\s$ the operators $\KC$ and $\KCbar $
 are explicit when acting on the standard basis.
 For this reason we will first establish all the necessary properties
 on the disc, $\Omega = \Db$, 
 and then translate them to general domains essentially by using the Riemann Mapping Theorem.

Define the orthonormal basis
$$
e_n(\theta) = (2\pi)^{-1/2} e^{in\theta} \in L^2(\s),
$$ 
in the standard parametrization of $\s$.
An explicit calculation yields,
\begin{equation}
 \KC e_n(\zeta) = \begin{cases}  (2\pi)^{-1/2} \zeta^n, & n \geq 0,\\ 0, & n<0,\end{cases} \label{eq : KC-on-monomials}
\end{equation}

and
\begin{equation*}
 \KCbar e_n(\zeta) = \begin{cases}  0, & n>0,\\(2\pi)^{-1/2} (\zeta^*)^{|n|}, & n \leq 0. 
\end{cases}
\end{equation*}

Furthermore for $L^2$-functions on the unit circle, we will denote the Fourier coefficients
\[
\widehat f (n) = \frac{1}{\sqrt{2 \pi}}\int_0^{2 \pi} f(\theta) e^{-in \theta} \id \theta = \Scalprod{e_n}{f}.
\]
We set for $s\in \R$
\[
\norm{f}_{H^s}^2 = \sum_{n \in \Z}(\abs{n}+1)^{2s} \abs{\widehat f(n)}^2 .
\]

The properties of $\KC$ and $\KCbar $ that we will need are grouped in the following proposition.
\begin{proposition} \label{prop : Properties-KC}
If $\Omega = \Db$ and $\KC$, $\KCbar $ are defined as above, then for all  $s \in [-1/2, 1/2]$ 
\begin{enumerate}[\it i)]
\item  $\KC$ and $\KCbar $ extend  to  bounded operators from $H^{-1/2}(\s)$ to $L^2(\Db)$. \label{prop : KC-bounded}
\item For all $f \in H^{s}(\s)$ we have $\partial_{z^*} \KC f = 0$ and $\partial_{z} \KCbar  f= 0$ 
with derivatives taken in the sense of distributions.\label{prop : KC-holomorphic}
\item  $\Trace \KC$ and $\Trace \KCbar $ extend to bounded operators  on $H^{s}(\s)$
 and  they are self-adjoint  projections onto ${\rm span}\,\{e_n  | n\ge 0\}$ and
 ${\rm span}\,\{e_n  | n\le  0\}$, respectively. \label{prop : trace-KC-bounded}
\item $\Trace \KC+ \Trace \KCbar  = 1 +\Scalprod{e_0}{\cdot}e_0 $ when acting on $H^s(\s)$. \label{prop : trace-KC-complementary} 
 \item  For  $\beta \in C^1(\s)$ and $s = -1/2$ or $s = 0$  the commutators $[\beta, \Trace \KC]$ and $[\beta, \Trace \KCbar ]$ are bounded from $H^{s}(\s)$ to $H^{s+1/2}(\s)$. 
 \label{prop : commutator-w-KC-improves}
\end{enumerate}

\end{proposition}

\begin{proof}
Point \ref{prop : trace-KC-complementary} is a direct consequence of \ref{prop : trace-KC-bounded}.
We will prove the remaining points for $\KC$ only, since the same
ideas apply to $\KCbar$.
Also, it is sufficient to establish these properties for continuous functions, since all statements extend to general elements of $H^s$ by density.

In this setting, point \ref{prop : KC-bounded} follows from \eqref{eq
  : KC-on-monomials}, since $\Scalprod{\zeta^n}{\zeta^k}=\frac{\pi}{n+1}
\delta_{n,k}$ and
\[
\norm{\KC f}_{L^2}^2 = \sum_{n,k \geq 0} \widehat f (n)^* \widehat f (k ) \Scalprod{\KC e_n}{\KC e_k} 
= \sum_{n \geq 0} \frac{1}{2n+2} \abs{\widehat f (n )}^2 
\leq  \norm{f}_{H^{-1/2}}^2.
\]
The proof of \ref{prop : KC-holomorphic} is straightforward.
Using \eqref{eq : KC-on-monomials} again we have that
$$
(\Trace \KC) e_n = \begin{cases}  e_n, & n \geq 0,\\ 0, & n<0,
\end{cases}
$$
which establishes point \ref{prop : trace-KC-bounded}.

To see \ref{prop : commutator-w-KC-improves}, we take $s = -1/2$ or $s = 0$, fix $f \in C^1 (\Bound)$ 
and compute the Fourier coefficients of $[\beta, \Trace \KC] f=
\beta\Trace \KC f-\Trace \KC(\beta f)$,

\[
\sqrt{2 \pi}\bigl([\beta, \Trace \KC] f\bigr)^{\wedge}(n) = \begin{cases}
                                                 \sum_{k \geq 0} \widehat \beta(n-k) \widehat f (k) - \sum_{k \in \Z} \widehat \beta(n-k) \widehat f(k), & n \geq 0, \\
                                                 \sum_{k \geq 0} \widehat \beta(n-k) \widehat f (k), & n < 0. 
                                                \end{cases}
\]
By Cauchy-Schwarz,
\begin{align*}
2 \pi& \abs{\bigl([\beta, \Trace \KC] f\bigr)^{\wedge}(n)}^2 \\
\qquad &\leq 
\begin{cases}
  \sum\limits_{k < 0} \abs{\widehat \beta(n-k)}^2 (\abs{k}+1)^{-2s}  \sum \limits_{k < 0} \abs{\widehat f (k)}^2 (\abs{k}+1)^{2s},  & n \geq 0, \\
  \sum\limits_{k \geq 0} \abs{\widehat \beta(n-k)}^2 (\abs{k}+1)^{-2s}
  \sum\limits_{k \geq 0} \abs{\widehat f (k)}^2 (\abs{k}+1)^{2s}, & n < 0,
\end{cases}\\
\qquad &\leq 
\|f\|_{H^s}^2
\begin{cases}
  \sum_{k < 0} \abs{\widehat \beta(n-k)}^2 (\abs{k}+1)^{-2s},   & n \geq 0 ,\\
  \sum_{k \geq 0} \abs{\widehat \beta(n-k)}^2 (\abs{k}+1)^{-2s}, & n < 0.
\end{cases}
\end{align*}
Therefore,  we obtain
\begin{align*}
 \norm{[\beta, \Trace \KC] f}_{H^{s+1/2}}^2  
	  &= \sum_{n \in \Z} (\abs{n}+1)^{2s+1} \abs{\bigl([\beta, \Trace \KC] f\bigr)^{\wedge}(n)}^2 \\
 	  & \leq  \norm{f}_{H^s}^2 \Bigl( 
 	  \sum_{\substack{n \geq 0 \\ k < 0}} (\abs{n}+1)^{2s+1} (\abs{k}+1)^{-2s} \abs{\widehat \beta(n-k)}^2 \\
	  & \quad + \sum_{\substack{n < 0 \\ k \geq 0}} (\abs{n}+1)^{2s+1} (\abs{k}+1)^{-2s} \abs{\widehat \beta(n-k)}^2 
 	  \Bigr).
\end{align*}
Since either $2s+1 = 0$ or $s = 0$ we get
\[
(\abs{n}+1)^{2s+1} (\abs{k}+1)^{-2s} \leq \abs{n} + \abs{k}+1 =  \abs{n-k}+1,
\]
where the last equality holds since $n $ and $k$ have opposite signs
in the sums we are considering.
This allows us to conclude that
\begin{align*}
 \norm{[\beta, \Trace \KC] f}_{H^{s+1/2}}^2  
  	    \leq & \norm{f}_{H^s}^2 \Bigl( 
 	  \sum_{\substack{n \geq 0 \\ k < 0}} (\abs{n-k}+1) \abs{\widehat \beta(n-k)}^2  \\
	&\qquad \qquad \qquad  + \sum_{\substack{n < 0 \\ k \geq 0}} (\abs{n-k}+1) \abs{\widehat \beta(n-k)}^2 
 	  \Bigr) \\
 	 \leq &  \norm{f}_{H^s}^2 \Bigl( 
 	  \sum_{m  \geq 1} (\abs{m}+1)^{2}  \abs{\widehat \beta(m)}^2 +\sum_{m \leq -1} (\abs{m}+1)^{2} \abs{\widehat \beta(m)}^2 
 	  \Bigr) \\
 	 \leq & \norm{f}_{H^s}^2 \norm{\beta}_{H^1}^2,
\end{align*}
which finishes the proof.
\end{proof}

The following lemma relates the operators $\KC$, $\KCbar$ and $S$ to
our problem at hand.

\begin{lemma} \label{lemma : S-improves-regularity}
Let $\Omega = \Db$ and  assume $v \in \mathcal{K} $. Then $\Trace S (\dotprod{\sigma}{n}\Trace v) \in H^{1/2}(\s, \C^2)$.
\end{lemma}
\begin{proof}
  Take a test function $f \in C^\infty(\s, \C^2)$ and a sequence $(v_n) \subset C^1(\overline{\Db}, \C^2)$ approaching $v$ in $\mathcal{K}$. 
 By Proposition  \ref{prop : Properties-KC} \ref{prop : trace-KC-bounded},
 $\Trace S$ is self-adjoint, thus using \eqref{eq : ibp} 
 \begin{align*}
   \int_{\s} \inprodS{\Trace S (\dotprod{\sigma}{n} \Trace v_n)}{f} &  = \int_{\s} \inprodS{ \Trace v_n}{ \dotprod{\sigma}{n}\Trace S f} \\
   & = -i \Scalprod{T v_n}{S f} + i \Scalprod{ v_n}{T S f} .
 \end{align*}
The last term  above cancels since, by Proposition \ref{prop :
  Properties-KC} \ref{prop : KC-holomorphic}, $T S f = 0$.
Thus, in view of Proposition \ref{prop : Properties-KC} \ref{prop :
  KC-bounded} 
\begin{align*}
\big|{ \int_{\s} \inprodS{\Trace S (\dotprod{\sigma}{n} \Trace v_n)}{f}}\big| &\leq \norm{T v_n}_{L^2(\Db)} \norm{S f}_{L^2(\Db)} \\
&\leq C_K \norm{T v_n}_{L^2(\Db)} \norm{ f}_{H^{-1/2}(\s)}.
\end{align*}
Taking the limit as $n \to \infty$ on both sides we see that $\Trace S (\dotprod{\sigma}{n}\Trace v)$ 
extends to a continuous functional on $H^{-1/2}$, and thus can be identified with a function in $H^{1/2}$.
\end{proof}
The next lemma allows us to conclude the proof of self-adjointness
when $\Omega = \Db$, see Remark \ref{lisa}. 
 \begin{lemma} \label{lemma : bootstrap}
Let $\Omega = \Db$ and  $\beta  $ be a nowhere vanishing $C^1(\s)$-function. 
Assume that $v \equiv \left(\begin{smallmatrix}   v_1 \\
    v_2   \end{smallmatrix}\right)  \in \mathcal{K}$ and that $\Trace
v_1 = \beta \Trace v_2$ as an equality in $H^{-1/2}(\s)$.
Then $\Trace v \in H^{1/2} (\s,\C^2)$.
\end{lemma}
\begin{remark}\label{lisa}
  In view of Lemma \ref{lemma : H1/2trace}, $v \in \Dom{D^*}$
  satisfies the hypotheses of Lemma \ref{lemma : bootstrap} with
  $\beta = \frac{t^* \cos \eta }{1- \sin \eta} $.  Thus, according to
  Lemma \ref{lemma : boundary->bulk}, $v\in H^1(\Db,\C^2)$ satisfies
  the boundary conditions. In particular, $\Dom{D^*}\subset \Dom{D}$.
\end{remark}
\begin{proof}
Let us write 
\begin{align*}
  \dotprod{\sigma}{n}=\begin{pmatrix}
0&n^*\\
n&0
\end{pmatrix}.
\end{align*}
In order to apply Lemma \ref{lemma : S-improves-regularity}, we define the spinor $f=
\dotprod{\sigma}{n} \gamma v$. Due to the boundary condition we have
that $f_2 = \tilde \beta f_1$ where $\tilde \beta = (n)^2 \beta$ is a $C^1$-function.
 In this notation Lemma \ref{lemma : S-improves-regularity} states that
 \begin{align}\label{eq:f12}
 \Trace \KC f_1 \in H^{1/2}, \quad \Trace \KCbar f_2\in H^{1/2}.
 \end{align}
Now we write
 \begin{align}\label{eq:matilde}
 \Trace \KC f_2 = \Trace \KC (\tilde \beta f_1) =  \tilde \beta \Trace\KC
 f_1 
 - [ \tilde \beta, \Trace \KC] f_1.
 \end{align}
 Clearly $\tilde \beta \Trace\KC f_1 $ is in $H^{1/ 2}$. By
 Proposition \ref{prop : Properties-KC} \ref{prop :
   commutator-w-KC-improves}, the term with the commutator is in
 $L^2$, so $\Trace \KC f_2 \in L ^2$ as well. This together with
 \eqref{eq:f12} gives that $f_2\in L^2$, in view of Proposition
 \ref{prop : Properties-KC} \ref{prop : trace-KC-complementary}.
 Since $\tilde \beta$ does not vanish, $f_1$ is also in $L ^2$ due to
 the boundary conditions.  With this improved regularity we return to
 \eqref{eq:matilde} and observe that, due to Proposition \ref{prop :
   Properties-KC} \ref{prop : commutator-w-KC-improves},  $[ \tilde \beta, \Trace \KC] f_1$
 is in $H^{1/2}$ so the same holds for $\Trace \KC f_2$.  Again using the complementarity of the projections and
 the fact $\tilde \beta $ does not vanish, we conclude $f_1, f_2 \in
 H^{1/2}$.
\end{proof}
\subsection{Riemann mapping and the proof of Theorem \ref{thm : self-adjoint} }

We first give the proof in the case where $\Omega$ is simply connected. 
The case of multiply connected domains will be treated  at the end of
this section.
Since $\Bound$ is $C^2$, there exists a $C^1$ conformal mapping (up to the boundary) $F: \overline{\Omega} \rightarrow \overline{\Db}$ 
with inverse $G$ \cite[Theorem 3.5, p. 48]{Pommerenkebook}.
Consider the map $U$ defined by $(Uf)(z):= f(G(z))$ mapping functions on $\overline{\Omega}$ to functions on $\overline{\Db}$.
By restriction (and abuse of notation), $U$ also maps functions on $\Bound$ to functions on $\s$.
\begin{lemma}\label{lem:Conformal}
When $\Omega$ is simply connected and has $C^2$-boundary, the map $U$ defines a bounded bijection from $L^2(\Omega)$ to $L^2(\Db)$ with bounded inverse. 
Furthermore,
$U: H^s(\Omega) \rightarrow H^s(\Db)$ is bounded with bounded inverse, for all $s \in [-1,1]$.

Similarly, $U: H^s(\Bound) \rightarrow H^s(\s)$ is bounded with bounded inverse, for all $s \in [-1,1]$.

Finally, if $ v\in \Dom{\D^*}$, then $Uv = (u_1, u_2) \in \mathcal{K}(\Db)$  
and on the boundary
$\Trace u_1 = \beta  \Trace u_2$ as an identity in $H^{-1/2}(\s)$, 
where $\beta = U( \frac{t^* \cos(\eta)}{1-\sin(\eta)} )$ is $C^1(\s)$.
\end{lemma}

\begin{proof}
Since $F, G$ have bounded derivatives on $\overline{\Omega}$ (resp $\overline{\Db}$) the map
$L^2(\Omega) \ni v \mapsto u := Uv = v \circ G \in L^2(\Omega)$ is a bounded bijection with bounded inverse. 
By direct differentiation one verifies that $U$ is also bounded from $H^1(\Omega)$ to $H^1(\Db)$ with bounded inverse. 
By interpolation and duality one finds that also $U: H^s(\Omega) \rightarrow H^s(\Db)$ is bounded with bounded inverse, for all $s \in [-1,1]$.

The same argument as in the interior applies on the boundary, so we
see that $U : H^s(\Bound) \rightarrow H^s(\s)$ is bounded with bounded inverse, 
for all $s \in [-1,1]$.

Suppose now that $v = (v_1,v_2) \in \Dom{\D^*}$. Then, since $\partial_{z^*} G = 0$, we have by the chain rule,
$$
\partial_{z} u_2 = G' \partial_{z} v_2  \in L^2(\Db), \qquad \partial_{z^*} u_1 = (G')^* \partial_{z^*} v_1  \in L^2(\Db).
$$
Finally, the boundary condition $\Trace u_1 = \beta  \Trace u_2$ follows from the boundary condition satisfied by $v$, see Lemma~\ref{lemma : H1/2trace}.
\end{proof}

Now we can conclude the proof of the self-adjointness of $\D$. 

\begin{proof}[Proof of Theorem \ref{thm : self-adjoint}]
{\it Simply connected case.}
Fix $v \in \Dom{D^*}$.
By Lemmas \ref{lemma : H1/2trace} and \ref{lemma : boundary->bulk}, we
only have to prove that $v$ has a well-defined trace in
$H^{1/2}(\Bound)$. By Lemma \ref{lem:Conformal}, this is  equivalent to
showing that $\Trace u := \Trace Uv \in H^{1/2}(\s)$, where
$U$ is the map defined above. 
By the same lemma, $u \in \mathcal{K}$ and its components $u_1, u_2$ satisfy the boundary condition
$\Trace u_1 = \beta \Trace u_2$, with $\beta= U( \frac{t^* \cos(\eta)}{1-\sin(\eta)} )$. 
Since $\beta$ vanishes nowhere by assumption, we can apply
Lemma~\ref{lemma : bootstrap} and conclude the proof of the theorem in
this case.\\
\noindent
{\it Multiply connected case.} 
It clearly suffices to consider connected $\Omega$.
Suppose that $\Bound$ is made up of the simple, regular curves
$\Gamma_0, \ldots, \Gamma_n$, with $n \geq 1$. Let $\Omega_j$ be the
interior  components of $\R^2 \setminus \Gamma_j$ (given
by the Jordan Curve Theorem). Let $\Gamma_0$ be the exterior boundary.
Since $\Omega$ is connected, $\Omega \subset \Omega_0$ and $\Omega \subset \R^2\setminus\overline{\Omega}_j$ for $j\geq 1$.

Let first $F_0: \Omega_0 \rightarrow \Db$ be the conformal (Riemann) map and let $U_0: L^2(\Omega) \rightarrow L^2(F_0(\Omega))$ be the push-forward map as in Lemma~\ref{lem:Conformal}. Proceeding as in the proof in the simply connected case, using $U_0$ instead of $U$, one concludes the desired $H^{1/2}$-regularity on the boundary component $\Gamma_0$.

For $j \in \{1,\ldots,n\}$ let $z_j \in \Omega_j$.
To obtain the $H^{1/2}$-regularity on the boundary component
$\Gamma_j$, one first applies the fractional transformation $I_j(z) =
(z-z_j)^{-1}$. After this transformation, $I_j(\Gamma_j)$ is the
external boundary of $I_j(\Omega)$ and one can proceed as in the
previous case. Notice that since $z_j \in \Omega_j$, the map $I_j$
(and its inverse) has bounded derivatives to all orders in $\Omega$
and therefore preserves Sobolev spaces in a similar manner to
Lemma~\ref{lem:Conformal}. This finishes the proof of Theorem \ref{thm : self-adjoint}.
\end{proof}
\noindent
{\bf Acknowledgments.}
The authors thank Martin Chuaqui for interesting discussions.
This work has  been
supported by the Iniciativa Cient\'ifica Milenio (Chile) through the
Millenium Nucleus RC–120002 ``F\'isica Matem\'atica'' . 
R.B.  has been supported by Fondecyt (Chile) Projects \# 116--0856 and
\#114--1155. S.F. acknowledges partial support from a Sapere Aude grant
from the Danish Councils for Independent Research, Grant number
DFF--4181-00221. E.S has been partially funded by Fondecyt (Chile)
project \# 114--1008. H. VDB. acknowledges support from Conicyt
(Chile) through CONICYT--PCHA/Doctorado Nacional/2014.  This work was
carried out while S.F. was invited professor at Pontificia Universidad
Cat\'olica de Chile.
\begin{appendix}

  \section{Construction of a Weyl sequence} \label{appendix2} In this
  appendix, we construct a singular Weyl sequence for $D_{\eta}$ at
  $0$ in the case where $\cos \eta$ vanishes to second order at a
  point of $\Bound$. This shows that if $D_{\eta}$ has a self-adjoint
  realisation then $0$ is in its essential spectrum. Therefore, by the
  compact embedding of $H^1$ in $L^2$, the domain of such a
  realisation cannot be included in $H^1$.

We will assume for definiteness that $\eta $ tends quadratically to $\frac{\pi}{2}$ at a point of the boundary. 
Then, the boundary conditions can be written $\Trace u_2 = B t \Trace u_1$, where $B =
(1 - \sin \eta)/ \cos \eta$.  Our assumption implies then that $\abs{B(s)} \leq
C (s-s_0)^2 $ and that $\abs{\dotprod{t}{}\nabla B(s)} \leq C \abs{s-s_0}$ for some $s_0
\in \Bound$.  We identify $\R^2$ with $\C$ and assume for definiteness
that $s_0 = 0$ and that $t(0) = i$.  Then, we can find $R_0 > 0$ such
that
\begin{equation}
 \Omega \cap \{re^{i \phi} | 0 \leq r \leq R_0, \abs{\phi} \leq \pi/4\} = \varnothing. \label{eq : outer_cone}
\end{equation}
Taking a smaller $R_0$ if necessary, we may also assume that 
$B$ and $t$ can be extended to $C^1$-functions on $\overline{\Omega} \cap B(0, R_0)$,
such that 
\[
\frac{\abs{B(z)}}{\abs{z}} + \abs{\nabla B (z)} \leq C_B \abs{z},
\quad \abs{t} + \abs{\nabla t} \leq C_t, \text{ for all } z \in
\overline{\Omega} \cap B(0, R_0),
\]
where $C_B$ and $C_t$ are positive constants.
This is always possible since such constants exist for $z \in \Bound$ and $\Omega$ has a $C^2$-boundary.
We also fix a cutoff function $\chi \in C^\infty(\R, [0,1])$ such that $\chi(x) = 1 $ for $x \leq 1/2$, $\chi(x) = 0$ for $x \geq 1$
and $\abs{\chi'} \leq 3 $.
For $R \geq 0$, define $\chi_R(z) = \chi(\abs{z}/R)$.

Now for $n \geq 1$, we set
\[
u_n(z) = (z-s_n)^{-n}\begin{pmatrix}
         1 \\ B t
        \end{pmatrix}
\]
for some $s_n > 0$.
Define $v_n := \chi_{R_n} u_n$. Notice that $v_n\in \Dom{D_\eta}$ for
all $R_n \leq R_0$ and we have that 
\[
\norm{v_n}  \geq \norm{\chi_{R_n}(z-s_n)^{-n}}.
\]
On the other hand,
\begin{align*}
 \norm{T v_n} &\leq \norm{\abs{\nabla \chi_{R_n}} u_n} + 2 \norm{\chi_{R_n}\begin{pmatrix}
                                                                                \partial_{z} B t(z-s_n)^{-n} \\
                                                                                \partial_{z^*} (z-s_n)^{-n}
                                                                               \end{pmatrix}}  \\
               &\leq \frac{3}{R_n} \norm{\mathbf{1}_{[R_n/2, R_n]}(|z|) u_n}
			 + 2 C_B C_t R_n
                         \norm{\chi_{R_n}(z-s_n)^{-n}}\\
&\quad +2 n C_t \norm{\chi_{R_n} B (z-s_n)^{-n-1}},
			 \end{align*}
where $\mathbf{1}_I$ is the indicator function on an intervall $I\subset\R$.
The last term can be estimated further by observing that, within $\supp{\chi_{R_n}} \cap \Omega$,
\[
\frac{\abs{B}}{\abs{z-s_n}} \leq C_B \frac{\abs{z}^2}{\abs{z-s_n}} \leq C_B R_n \sqrt{2},
\]
where the last inequality holds in view of \eqref{eq : outer_cone}.
Thus, we obtain 
\[
\frac{\norm{T v_n}}{ \norm{v_n}} \leq \frac{3}{R_n} \frac{\norm{\mathbf{1}_{[R_n/2, R_n]}(|z|) u_n}}{\norm{\chi_{R_n}(z-s_n)^{-n}}}
	+ C_B C_t (2 + \sqrt{2} n ) R_n. 
\]
We now fix $R_n \leq R_0$ such that the second term in the above equation is smaller than $1/ 2 n$.
In the first term, we note that, as $s_n \searrow 0$ for a fixed $R_n$, the numerator stays bounded while the 
denominator increases to $+ \infty$. Thus, by choosing a sufficiently small $s_n$,
we obtain
\[
\frac{\norm{T v_n}}{ \norm{v_n}}\leq \frac{1}{n}.
\]
In addition, the sequence $v_n/ \norm{v_n}$ converges weakly to zero, so it is a singular Weyl sequence, which proves $0 \in \sigma_{\rm ess}(\D_\eta)$. 

\end{appendix}

\providecommand{\bysame}{\leavevmode\hbox to3em{\hrulefill}\thinspace}
\providecommand{\MR}{\relax\ifhmode\unskip\space\fi MR }
\providecommand{\MRhref}[2]{%
  \href{http://www.ams.org/mathscinet-getitem?mr=#1}{#2}
}
\providecommand{\href}[2]{#2}

\end{document}